\documentclass[11pt]{article}
\usepackage{graphicx}
\graphicspath{{images/} }
\makeatletter
\def\temp{dvips.def}
\ifx\Gin@driver\temp
\def\Ginclude@graphics#1{\def\temp{#1}---image \expandafter\strip@prefix\meaning\temp---}
\fi
\makeatother

\usepackage{amsmath}
\usepackage{amsthm}
\usepackage{amssymb}
\usepackage{amsfonts}
\usepackage{subfigure}
\usepackage{color}
\usepackage{makeidx}

\newcommand{\ignore}[1]{}

\newtheorem{theorem}{Theorem}

\newtheorem{lemma}{Lemma}
\newtheorem{definition}[theorem]{Definition}

\newcommand{\Min}{{\rm Min}}
\newcommand{\Next}{{\rm Next}}

\newcommand{\FF}{\mathbb{F}}

\begin{document}

\title{Enumerating
all the Irreducible Polynomials\\ over Finite Field}
\author{{\bf Nader H. Bshouty}\\ Dept. of Computer Science\\ Technion\\  Haifa, 32000
\and {\bf Nuha Diab} \\ Sisters of Nazareth High School \\ Grade 12\\ P.O.B. 9422, Haifa, 35661
\and {\bf Shada R. Kawar}  \\ Nazareth Baptist High School\\ Grade 11\\ P.O.B. 20, Nazareth, 16000
\and {\bf Robert J. Shahla} \\ Sisters of Nazareth High School \\
Grade 11\\ P.O.B. 9422, Haifa, 35661}


%
\maketitle

\begin{abstract}
In this paper we give a detailed analysis of deterministic and randomized algorithms that enumerate any number of irreducible polynomials of degree $n$
over a finite field and their roots in the extension field
in quasilinear\footnote{$O(N\cdot poly(\log N))$ where $N=n^2$ is the size of the output.} time cost per element.

Our algorithm is based on an improved algorithm for enumerating all the Lyndon words of length $n$ in linear delay time and the known reduction of Lyndon words to irreducible polynomials.
\end{abstract}

\section{Introduction}
The problem of enumerating the strings in a language $L$ is to
list all the elements in $L$ in some order.
Several papers study this problem. For example,
Enumerating all spanning trees, \cite{KR00}, minimal transversals for some Geometric Hypergraphs,
\cite{EMR09}, maximal cliques, \cite{MU04}, ordered trees, \cite{E85},
certain cuts in graphs, \cite{VY92,YWS10}, paths in a graph, \cite{S95},
bipartite perfect matchings, \cite{U01}, maximum and maximal matchings in bipartite graphs, \cite{U97},
and directed spanning trees in a directed graph~\cite{U96}.
See the list in \cite{FM} for other enumeration problems.

One of the challenges in enumeration problems
is to find an order of the elements of $L$
such that finding the next element in that order
can be done in quasilinear time in the length
of the representation of the element. The time that the algorithm
takes before giving the first element
is called the {\it preprocessing time}. The time of
finding the next element is called the {\it delay time}.
In~\cite{AS09}, Ackerman and Shallit gave a linear preprocessing
and delay time for enumerating
the words of any regular language (expressed as a regular expression or NFA)
in lexicographic order.

Enumeration is also of interest to mathematicians without
addressing the time complexity.  Calkin and Wilf,\cite{CW00}, gave an
enumeration of all the rational numbers such that
the denominator of each fraction is the numerator of the next one.

Another problem that has received considerable attention is the problem
of ranking the elements of $L$. In ranking the goal is to find
some total order on the elements of $L$ where the problem of
returning the $n$th element in that order can be solved in polynomial time. Obviously, polynomial
time ranking implies polynomial time enumeration.
In the literature, the problem of ranking is already solved for permutations~\cite{MR01,T08}
and trees of special properties \cite{GLW82,L86,P86,RW11,WC11,WCW11,WCC11,AKN11,WCCL13,WCCK13}.
Those also give enumerating algorithms for such objects.

Let $\FF_q$ be a finite field with $q$ elements.
Let $P_{n,q}$ be the set of irreducible polynomials over $\FF_q$
of degree $n$ and their roots in $\FF_{q^n}$. Several algorithms
in the literature use irreducible polynomials of degree $n$ over finite fields, especially algorithms in coding theory, cryptography and problems that use the Chinese Remainder Theorem for polynomials~\cite{C71,LSW,B15,DPS}.
Some other algorithms use only the roots of those polynomials. See for example~\cite{B15}.

In this paper, we study the following problems
\begin{enumerate}
\item Enumeration of any number of
irreducible polynomials of degree $n$ over a finite fields.
\item Enumeration of any number of
irreducible polynomials of degree $n$ and their roots over the extended field.
\item Enumeration of any number of roots of irreducible polynomials of degree $n$ over the extended field. One root for each polynomial.
\end{enumerate}

There are many papers in the literature that mention the result of enumerating all the irreducible polynomials of degree {\it less than or equal} to $n$ but do not give the exact algebraic complexity of this problem~\cite{CR00,D88,RS92,FK86,FM78,KRR15}. In this paper we give a detailed analysis of deterministic and randomized algorithms that enumerate any number of irreducible polynomials of degree $n$
over a finite field and/or their roots in the extension field
in quasilinear\footnote{$O(N\cdot poly(\log N))$ where $N=n^2$ is the size of the output.} time cost per element.

Our algorithm is based on an improved algorithm for enumerating all the Lyndon words of length $n$ in linear delay time and the well known reduction of Lyndon words to irreducible polynomials. In the next subsection we define the Lyndon word and present the result of the improved algorithm.

\subsection{The Enumeration of Lyndon Words}
Let $<$ be any total order on $\FF_q$.
A {\it Lyndon word} (or string) over $\FF_q$ of length $n$ is a word $w=w_1\cdots w_n\in \FF_q^n$
where every rotation $w_i\cdots w_nw_1\cdots w_{i-1}$, $i\not=1$ of $w$ is lexicographically larger than $w$. Let $L_{n,q}$ be the set of all the Lyndon words over $\FF_q$ of length $n$.
In many papers in the literature, it is shown that there is
polynomial time (in $n$) computable bijective function $\phi: L_{n,q}\to P_{n,q}$, where $P_{n,q}$ is the set of all polynomials of degree $n$ over $\FF_q$.
So the enumeration problem of the irreducible polynomials can be reduced to the problem of enumerating
the elements of $L_{n,q}$.

Bshouty gave in ~\cite{B15} a large subset $L'\subseteq L_{n,q}$
where any number of words in $L'$ can be enumerated in a linear delay time.
In fact, one can show that $L'$
has a small DFA and, therefore, this result follows from~\cite{CW00}.
It is easy to show that the set
$L_{n,q}$ cannot be accepted by a small size NFA, i.e.,
size polynomial in $n$, so one cannot generalize the above result to all $L_{n,q}$.
Duval~\cite{D88} and Fredricksen et. al., \cite{FK86,FM78} gave
enumeration algorithms of all the words in $\cup_{m\le n}L_{m,q}$ that run in linear delay time.
Berstel and Pocchiola in~\cite{BP94} and Cattell et. al. in~\cite{CR00,RS92} show that, in Duval's algorithm,
in order to find the next Lyndon word in $\cup_{m\le n}L_{m,q}$, the amortized number of {\it updates} is constant.
The number of updtes is the number of symbols that the algorithm change in a Lyndon word in order to get the next word.
Such an algorithm is called CAT algorithm. See the references in \cite{CR00}
for other CAT algorithms.
Kociumaka et. al. gave an algorithm that finds the rank of a Lyndon
word in $O(n^2\log q)$ time and does unranking in $O(n^3\log^2 q)$ time.

In this paper, we give an enumeration algorithm of $L_{n,q}$ with linear delay time. Our algorithm is the same as Duval's algorithm with the addition of a simple data structure. We show that this data structure
enable us to find the next Lyndon word of {\it length $n$} in constant updates per symbol and therefore in linear time. We also show that our algorithm is CAT algorithm and give an upper bound for the amortized update cost.

Another problem is testing whether a word of length $n$ is Lyndon word. In~\cite{D83}, Duval gave a linear time algorithm for such test. In this paper we give a simple algorithm that uses the suffix trie data structure and runs in linear time.

This paper is organized as follows. In Section 2 we give the exact arithmetic complexity of the preprocessing and delay time for enumerating any number of irreducible polynomials and/or their roots. In Section 3 we give a simple data structure that enable us to change Duval's algorithm to an algorithm that enumerates all the Lyndon words of length $n$ in linear delay time. We then show in Section 4 that the algorithm is CAT algorithm. In Section 5 we give a simple linear time algorithm that tests whether a word is a Lyndon word.

\section{Enumerating Irreducible Polynomials}
In this section we give the analysis for the algebraic complexity of the preprocessing time and delay time of enumerating irreducible polynomials of degree $n$ over a finite field and/or their roots in the extended field.

Let $q$ be a power of a prime $p$ and $\FF_{q}$ be the finite field with $q$ elements.
Our goal is to enumerate all the irreducible polynomials of degree $n$
over $\FF_q$ and/or their roots in the extension field $\FF_{q^n}$.

The best deterministic algorithm for constructing an irreducible polynomial over $\FF_q$
of degree $n$ has time complexity $T_D:=O(p^{1/2+\epsilon}n^{3+\epsilon}+(\log q)^{2+\epsilon}n^{4+\epsilon})$
for any $\epsilon>0$. The best randomized algorithm has time complexity $T_R:=O((\log n)^{2+\epsilon} n^{2}+(\log q) (\log n)^{1+\epsilon} n)$
for any $\epsilon>0$. For a comprehensive survey of this problem see~\cite{S99}
Chapter 3. Obviously, the preprocessing time for enumerating irreducible polynomials cannot be less than the time for constructing one.
Therefore, $T_D$ for the deterministic algorithm, and $T_R$ for the randomized algorithm.

The main idea of the enumeration algorithm is to enumerate the roots of the irreducible polynomials
in the extension field and then construct the polynomials from their roots.
Let $\FF_{q^n}$ be the extension field of $\FF_q$ of size $q^n$.
One possible representation of the elements of the field $\FF_{q^n}$ is by polynomials of degree
at most $n-1$ in $\FF_q[\beta]/(f(\beta))$ where $f(x)$ is an irreducible polynomial of degree $n$.
A {\it normal basis} of $\FF_{q^n}$ is a basis over $\FF_q$ of the form
$N(\alpha):=\{\alpha,\alpha^q,\alpha^{q^2},\ldots,\alpha^{q^{n-1}}\}$ for some $\alpha\in \FF_{q^n}$ where $N(\alpha)$ is linearly independent.
The {\it normal basis theorem} states that for every finite field $\FF_{q^n}$
there is a normal basis $N(\alpha)$. That is, an $\alpha$ for which $N(\alpha)$ is
linearly independent over $\FF_q$.
It is known that such an $\alpha$ can be constructed
in deterministic time $O(n^3+(\log n)(\log\log n)(\log q) n)$ and randomized time
$O((\log \log n)^2(\log n)^4 n^2 + (\log n)(\log \log n)(\log q) n )$~\cite{GS,Nor01,Nor02}.
The enumeration algorithm will use the normal basis for representing the elements of $\FF_{q^n}$.
Notice that the time complexity to find such an element $\alpha$ is less than constructing one irreducible polynomial. If we use the normal basis
$N(\alpha)$ for the representation of the elements of $\FF_{q^n}$, then
every element $\gamma\in \FF_{q^n}$ has a unique representation $\gamma=\lambda_1\alpha+\lambda_2\alpha^{q}+\lambda_3\alpha^{q^2}+\cdots+\lambda_n\alpha^{q^{n-1}}$
where $\lambda_i\in \FF_q$ for all $i$.

It is known that any irreducible polynomial $g$ of degree $n$ over $\FF_q$
has $n$ distinct roots in $\FF_{q^n}$.
If one can find one root $\gamma\in \FF_{q^n}$ of $g$ then the other roots are $\gamma^q,\gamma^{q^2},\ldots,\gamma^{q^{n-1}}$ and therefore
$g_\gamma(x):=(x-\gamma)(x-\gamma^q)\cdots (x-\gamma^{q^{n-1}})=g(x)$.
The coefficients of $g_\gamma(x)$ can be computed in quadratic time $O(n^2\log^3 n(\log\log n)^2)$.
See Theorem~A and~B in~\cite{S99} and references within.
The element $\gamma=\lambda_1\alpha+\lambda_2\alpha^{q}+\lambda_3\alpha^{q^2}+\cdots+\lambda_n\alpha^{q^{n-1}}$
is a root of an irreducible polynomial of degree $n$ if and only if $\gamma,\gamma^{q},\gamma^{q^2},\ldots,\gamma^{q^{n-1}}$ are distinct. Now since
\begin{eqnarray}
\gamma^{q^{n-k}}=\lambda_k\alpha+\lambda_{k+1}\alpha^q\cdots+
\lambda_n\alpha^{q^{n-k}}+\lambda_1\alpha^{q^{n-k+1}}+\cdots+\lambda_{k-1}\alpha^{q^{n-1}},
\label{pk}
\end{eqnarray}
$\gamma$ is a root of an irreducible polynomial of degree $n$ if and only if
the following $n$ elements
$$(\lambda_1,\lambda_2,\lambda_3,\cdots,\lambda_n), (\lambda_2,\lambda_3,\lambda_4,\cdots,\lambda_n,\lambda_1)
,(\lambda_3,\lambda_4,\lambda_5,\cdots,\lambda_n,\lambda_1,\lambda_2),\cdots,$$
\begin{eqnarray}(\lambda_n,\lambda_1,\lambda_2,\cdots,\lambda_{n-1})\label{ddd}\end{eqnarray}
are distinct.

When (\ref{ddd}) happens then we call $\lambda=(\lambda_1,\lambda_2,\lambda_3,\cdots,\lambda_n)$
{\it aperiodic word}.
We will write $\lambda$ as a word $\lambda=\lambda_1\lambda_2\lambda_3\cdots\lambda_n$ and define $\gamma(\lambda):=\lambda_1\alpha+\lambda_2\alpha^{q}+\lambda_3\alpha^{q^2}+\cdots+\lambda_n\alpha^{q^{n-1}}$.
Therefore
\begin{lemma} We have
\begin{enumerate}
\item For any word $\lambda=\lambda_1\cdots\lambda_n\in \FF_q^n$ the element $\gamma(\lambda)$ is a root of an irreducible polynomial of degree $n$ if and only if $\lambda$ is an aperiodic word.
\item Given an aperiodic word $\lambda$, the irreducible polynomial $g_{\gamma(\lambda)}$ can be constructed in time\footnote{Here $\tilde O(N)=\tilde O(N\cdot poly(\log (N)))$} $O((\log\log n)^2(\log n)^3 n^2)=\tilde O(n^2)$.
    \end{enumerate}
\end{lemma}
Obviously, the aperiodic word $\lambda=\lambda_1\lambda_2\lambda_3 \cdots\lambda_n$
and $R_k(\lambda):=\lambda_k\lambda_{k+1}$ $\cdots\lambda_n\lambda_1$ $\cdots\lambda_{k-1}$
corresponds to the same irreducible polynomial. See~(\ref{pk}). That is, $g_{\gamma(\lambda)}=g_{\gamma(R_i(\lambda))}$
for any $1\le i\le n$.
Therefore to avoid enumerating the same polynomial
more than once, the algorithm enumerates only the minimum element (in lexicographic order)
among $\lambda,R_2(\lambda),\ldots,R_{n}(\lambda)$. Such an element is called
{\it Lyndon word}.
Therefore
\begin{definition} The word $\lambda=\lambda_1\lambda_2\lambda_3\cdots\lambda_n$ is called a Lyndon word if $\lambda<R_i(\lambda)$ for all $i=2,\ldots,n$.
\end{definition}
To enumerate all the irreducible polynomials the algorithm enumerates all the Lyndon words of length $n$ and, for each one, it computes the corresponding irreducible polynomial.

\begin{figure}
\centering
\includegraphics[trim = 0 0cm 0 0cm,width=1.1\textwidth]{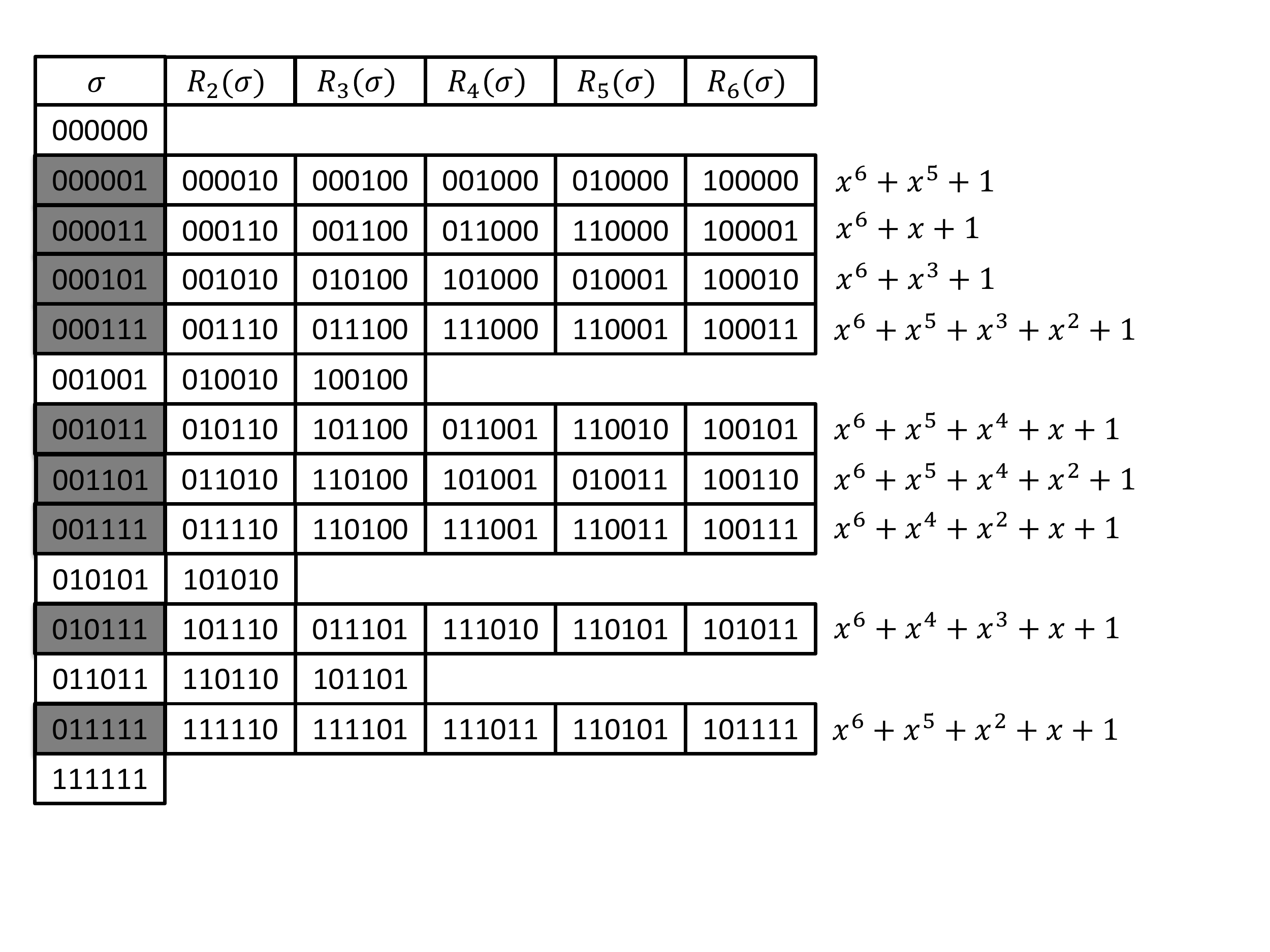}
\caption{A table of the words over $\Sigma=\{0,1\}$ and all their rotations.
The Lyndon words are in the gray boxes.
The Lyndon words of length $6$ are $000001$, $000011$, $000101$, $000111$, $001011$, $001101$, $001111$, $010111$ and $011111$. The polynomial $f(x)=x^6+x+1$ is irreducible over $\FF_2$ and therefore $\FF_{2^6}=\FF_2[\beta]/(\beta^6+\beta+1)$ and every element in $\FF_{2^6}$ can be represented as $\lambda_5\beta^5+\cdots+\lambda_1\beta+\lambda_0$.
For $\alpha=\beta^5+\beta^2+1$ the set $N(\alpha)=\{\alpha,\alpha^2,\alpha^4,\alpha^8,\alpha^{16},\alpha^{32}\}$ is a Normal basis.
The Lyndon word $001011$ corresponds to the element $\gamma=\alpha^4+\alpha^{16}+\alpha^{32}$. The element $\gamma$ corresponds to the irreducible
polynomial $g_\gamma(x)=(x-\gamma)(x-\gamma^2)(x-\gamma^4)(x-\gamma^8)(x-\gamma^{16})(x-\gamma^{32})$ $=x^6+x^5+x^4+x+1$.}
\label{n6}
\end{figure}

In the next section, we show how to enumerate all the Lyndon words of length $n$ in linear delay time
$O(n)$.
Then from $\gamma(\lambda)$ (that corresponds to an irreducible polynomial)
the algorithm constructs the irreducible polynomial $g_{\gamma(\lambda)}(x)$
and all the other $n-1$ roots in quadratic time~$\tilde O(n^2)$. Since the
size of all the roots is $O(n^2)$, this complexity is quasilinear in the output size.
For the problem of enumerating only the roots (one root for each irreducible polynomial) the delay time is $O(n)$.

Let $L_{n,q}$ be the set of all Lyndon words over $\FF_q$ of length~$n$.
We have shown how to reduce our problem to the problem of enumerating all the Lyndon words over $\FF_q$ of length $n$ with linear delay time. Algorithm ``Enumerate'' in Figure~\ref{Alg} shows the reduction.

\begin{figure}[h!]
  \begin{center}
  \fbox{\fbox{\begin{minipage}{28em}
  \begin{tabbing}
  xxxx\=xxxx\=xxxx\=xxxx\= \kill
  {\bf Enumerate$(n,q)$}\\
  Preprocessing\\
  \> 1p) Find an irreducible polynomial $f(x)$ of degree $n$ over $\FF_q$.\\
  \> 2p) Find a normal basis $\alpha,\alpha^q,\ldots,\alpha^{q^{n-1}}$ in $\FF_q[\beta]/(f(\beta))$.\\
  \> 3p) Let $\lambda=00\cdots01$ \ \  /* The first Lyndon word */  \\
  Delay\\
  \> 1d) Define $\gamma=\lambda_1\alpha+\lambda_2\alpha^q+\cdots+\lambda_n\alpha^{q^{n-1}}$.\\
  \> 2d) Compute $g_\gamma(x):=(x-\gamma)(x-\gamma^q)\cdots (x-\gamma^{q^{n-1}})$.\\
  \> 3d) Output($g_\gamma(x)$, $\gamma,\gamma^q,\cdots,\gamma^{q^{n-1}}$).\\
  \> 4d) Find the next Lyndon word: $\lambda\gets \Next(\lambda)$.\\
  \> 5d) If $\lambda=00\cdots01$ then Halt else Goto 1d.
  \end{tabbing}
  \end{minipage}}}
  \end{center}
	\caption{An enumeration algorithm.}
	\label{Alg}
	\end{figure}

Putting all the above algebraic complexities together, we get the following
\begin{theorem} Let $\epsilon>0$ be any constant. There is a randomized enumeration algorithm for
\begin{enumerate}
\item the irreducible polynomial over $\FF_q$ and their roots in $\FF_{q^n}$ in preprocessing time $O((\log n)^4(\log\log n)^2 n^2+(\log q)(\log n)^{1+\epsilon}n)$ and delay time $O((\log\log n)^2(\log n)^3n^2)$.
\item the roots in $\FF_{q^n}$ of irreducible polynomials of degree $n$ over $\FF_q$ in preprocessing time $O((\log n)^4(\log\log n)^2 n^2+(\log q)(\log n)^{1+\epsilon}n)$ and delay time $O(n)$.
\end{enumerate}
\end{theorem}

\begin{theorem} Let $\epsilon>0$ be any constant. There is a deterministic enumeration algorithm for
\begin{enumerate}
\item the irreducible polynomial over $\FF_q$ and their roots in $\FF_{q^n}$ in preprocessing time $O(n^{3+\epsilon}p^{1/2+\epsilon}+(\log q)^{2+\epsilon}n^{4+\epsilon}$ and delay time $O((\log\log n)^2$ $(\log n)^3n^2)$.
\item the roots in $\FF_{q^n}$ of irreducible polynomials of degree $n$ over $\FF_q$ in preprocessing time $O(n^{3+\epsilon}p^{1/2+\epsilon}+(\log q)^{2+\epsilon}n^{4+\epsilon})$ and delay time $O(n)$.
\end{enumerate}
\end{theorem}

\section{Linear Delay Time for Enumerating $L_{n,q}$}
In this section we give Duval's algorithm,~\cite{D88}, that enumerates all the Lyndon words of length at most $n$, $\cup_{m\le n} L_{m,q}$, in linear delay time and change it to an algorithm that enumerates the Lyndon words of length $n$, $L_{n,q}$ in linear time. We will use a simple data structure that enable the algorithm to give the next Lyndon word of length $n$ in Duval's algorithm in a constant update per symbol and therefore in linear time.

Let $\Sigma=\{0,1,\ldots , q-1\}$ be the alphabet with the order $0<1<\cdots <q-1$. We here identify $\FF_q$ with $\Sigma$. We will sometime write the symbols in brackets. For example for $q=5$ the word $[q-1]^2[q-3]$ is $442$.  Let $w=\sigma_1\sigma_2\cdots\sigma_m$ be a Lyndon word for some $m\le n$. To find the next Lyndon word, (of length $\le n$) Duval's algorithm first define the word $v=D(w)=w^hw'$ of length $n$ where $w$ is a non-empty prefix of $w$ and $h\ge 0$ (and therefore $h|w|+|w'|=n$). That is, $v=D(w)=\sigma_1\cdots\sigma_m\sigma_1\cdots\sigma_m\cdots \sigma_1\cdots\sigma_m\sigma_1\cdots\sigma_{(n\mod m)}$. Then if $v$ is of the form $v=ub[q-1]^t$ where $t\ge 0$ and $b\not= [q-1]$ then the next Lyndon word in Duval's algorithm is $P(v)=u[b+1]$. We denote the next Lyndom word of $w$ (in Duval's algorithm) by $N(w):=P(D(w))$.
For example, for $q=3$, $n=7$ and $w=0222$, $v=D(w)=0222022$ and $N(w)=P(D(w))=02221$. Then $N(N(w))=022211$.

The following lemma is well known. We give the proof for completeness
\begin{lemma}\label{increase}
If $w$ is a Lyndon word and $|w|<n$ then $|N(w)|>|w|$.
\end{lemma}
\begin{proof}
Let $w=ub[q-1]^t$ where $b\not= [q-1]$. Then $u_1\le b$ because otherwise we would have $R_{|u|+1}(w)=b[q-1]^tu<ub[q-1]^t=w$
and then $w$ is not a Lyndon word.
Let $D(w)=w^hw'$ where $h\ge 0$ and $w'$ is a nonempty prefix of $w$. Since $|D(w)|=n>|w|$ we have $h\ge 1$. Since $w'_1=u_1\le b<q-1$, we have that
$|N(w)|=|P(D(w))|\ge h|w|+1>|w|.$
\end{proof}

\subsection{The Algorithm}
In this subsection we give the data structure and the algorithm that finds the next Lyndon word of length $n$ in linear time.

We note here that, in the literature, the data structure that is used for the Lyndon word is an array of symbols. All the analyses of the algorithms in the literature treat an access to an element in an $n$ element array and comparing it to another symbol as an operation of time complexity equal to~$1$. The complexity of incrementing/decrementing an index $0\le i\le n$ of an array of length $n$ and comparing two such indices are not included in the complexity. In this paper, the Lyndon words are represented with symbols and {\it numbers} in the range $[1,n]$. Every access to an element in this data structure and comparison between two elements are (as in literature) counted as an operation of time complexity equal $1$. Operations that are done on the indices of the array (as in literature) are not counted but their time complexity is linear in the number of updates.

Let $v\in \Sigma^n$. We define the {\it compressed representation} of $v$ as $v=v^{(0)}[q-1]^{i_1}v^{(1)}[q-1]^{i_2}\cdots v^{(t-1)}[q-1]^{i_t}$ where $i_1,\ldots,i_{t-1}$ are not zero ($i_t$ may equal to zero) and $v^{(0)},\ldots,v^{(t-1)}$ are nonempty words that do not contain the symbol $[q-1]$. If $v$ do not contain the symbol $[q-1]$ then $v=v^{(0)}[q-1]^0$ where $[q-1]^0$ is the empty word and $v^{(0)}=v$.
The data structure will be an array (or double link list) that contains $v^{(0)},i_1,v^{(1)},\cdots,v^{(t-1)},i_t$ if $i_t\not=0$ and $v^{(0)},i_1,v^{(1)},\cdots,v^{(t-1)}$ otherwise.

Define $\|v\|=\sum_{j=0}^{t-1} |v^{(j)}|+t$. This is the {\it compressed length} of the compressed representation of $v$.
Notice that for a word $v=v_1\cdots v_r$ that ends with a symbol $v_r\not=[q-1]$
we have $P(v)=v_1\cdots v_{r-1}[v_r+1]$ and for $u=v\cdot [q-1]^i$ we have $P(u)=P(v)$. Therefore $\|v\|-1\le \|P(v)\|\le \|v\|$.

Let $v=v^{(0)}[q-1]^{i_1}v^{(1)}[q-1]^{i_2}\cdots v^{(t-1)}[q-1]^{i_t}$  be any Lyndon word of length $n$.
The next Lyndon word in Duval's algorithm is
$$u^{(1)}:=N(v)=v^{(0)}[q-1]^{i_1}v^{(1)}[q-1]^{i_2}\cdots [q-1]^{i_{t-1}}\cdot P(v^{(t-1)})$$
To find the next Lyndon word $u^{(2)}$ after $u^{(1)}$ we take $\left(u^{(1)}\right)^hz^{(1)}$ of length $n$ where $z^{(1)}$ is a nonempty prefix of $u^{(1)}$ and then $u^{(2)}=\left(u^{(1)}\right)^h\cdot P(z^{(1)})$.
This is because $z^{(1)}_1=u^{(1)}\not=[q-1]$.
Since by Lemma~\ref{increase}, $|u^{(1)}|<|u^{(2)}|<\cdots$ we will eventually get a Lyndon word of length $n$. We now show that using the compressed representation we have
\begin{lemma} \label{kkk}The time complexity of computing $u^{(i+1)}$ from $u^{(i)}$ is at most $|u^{(i+1)}|-|u^{(i)}|+1$.
\end{lemma}
\begin{proof} Let $u^{(i)}=w^{(0)}[q-1]^{i_1}w^{(1)}[q-1]^{i_2}\cdots w^{(t-1)}[q-1]^{i_t}$ of length less than $n$. Then
$u^{(i+1)}=(u^{(i)})^h\cdot P(z^{(i)})$ where $z^{(i)}$ is a nonempty prefix of $u^{(i)}$. So it is enough to show that $P(z^{(i)})$ can be computed in at most $|P(z^{(i)})|+1$ time. Notice that the length of $z^{(i)}$ is $(n\mod |u^{(i)}|)$ (here the mod is equal to $|u^{(i)}|$ if $|u^{(i)}|$ divides $n$). Since $z^{(i)}$ is a prefix of $u^{(i)}$ we have that,
in the compressed representation, $z^{(i)}=w^{(0)}[q-1]^{i_1}w^{(1)}[q-1]^{i_2}\cdots w^{(t'-1)}[q-1]^{i_{t'}}$ for some $t'\le t$. Then
$P(z^{(i)})=w^{(0)}[q-1]^{i_1}w^{(1)}[q-1]^{i_2}\cdots P(w^{(t'-1)})$. Therefore the complexity of computing $P(z^{(i)})$ is $\|z^{(i)}\|\le |P(z^{(i)})|+1.$
\end{proof}

From the above lemma it follows that
\begin{theorem}\label{Th3} Let $v$ be a Lyndon word of length $n$. Using the compressed representation, the next Lyndon word of length $n$ can be computed in linear time.
\end{theorem}
\begin{proof} To compress $v$ and find $u^{(1)}=N(v)$ we need a linear time. By Lemma~\ref{increase} the Lyndon words after $v$ are $u^{(1)},\ldots, u^{(j)}$ where $|u^{(1)}|<|u^{(2)}|<\cdots<|u^{(j)}|=n$. By Lemma~\ref{kkk} the time complexity of computing the next Lyndon word $u^{(j)}$ of length $n$ is $\sum_{i=1}^{j-1} |u^{(i+1)}|-|u^{(i)}|+1\le |u^{(j)}|+n=O(n)$. Then decompressing the result takes linear time.
\end{proof}

We now give a case where Duval's algorithm fails to give the next Lyndon word of length $n$ in linear time. Consider the Lyndon word $01^{k}01^{k+1}$ of length $n=2k+3$. The next Lyndon word in Duval's algorithm is $01^{k+1}$. Then $01^{k+2}, 01^{k+3}, \ldots, 01^{2k+2}$. To get to the next Lyndon word of length $n$, $01^{2k+2}$, the algorithm does $\sum_{i=1}^{k+2}i=O(n^2)$ updates.

\section{Constant Amortized Time for Enumerating $L_{n,q}$}
In this section, we show that our algorithm in the previous section is CAT algorithm. That is, it has a constant amortized update cost. 

We first give some notation and preliminary results.
Let $\ell_n$ be the number of Lyndon words of length $n$, $L_i=\ell_1+\cdots +\ell_i$ for all $i=1,\ldots, n$ and $\Lambda_n=L_1+\cdots+L_n=n\ell_1+(n-1)\ell_{2}+\cdots+\ell_n$. It is known from~\cite{D88} that for $n\ge 11$ and any $q$
\begin{eqnarray}\label{one}
\frac{q^n}{n}\left(1-\frac{q}{(q-1)q^{n/2}}\right)\le \ell_n\le \frac{q^n}{n}
\end{eqnarray}
and for any $n$ and $q$
\begin{eqnarray}\label{two}
L_n\ge \frac{q}{q-1}\frac{q^n}{n}
\end{eqnarray}
and
\begin{eqnarray}\label{three}
\Lambda_n= \frac{q^2}{(q-1)^2}\frac{q^n}{n}\left(1+\frac{2}{(q-1)(n-1)}+O\left(\frac{1}{(qn)^2}\right)\right).
\end{eqnarray}
Denote by $\ell_{n,i}$ the number of Lyndon words of length $n$ of the form $w=ub[q-1]^i$ where $b\in \Sigma\backslash \{q-1\}$. Then $\ell_n=\ell_{n,0}+\ell_{n,1}+\cdots+\ell_{n,n-1}$.
Let $\ell^*_n$ be the number of Lyndon words of length $n$ that ends with the symbol $[q-2]$.
That is, of the form $u[q-2]$.

For the analysis we will use the following.
\begin{lemma} Let $w=ub[q-1]^t\in \Sigma^n$ where $b\in \Sigma\backslash \{q-1\}$ and $t\ge 1$. If $w=ub[q-1]^t$ is a Lyndon word of length $n$ then $u[b+1]$ is a Lyndon word.

In particular,
$$\ell_{n,t}\le \ell_{n-t}.$$

If $w=u[q-2]$ is a Lyndon word of length $n$ then $u[q-1]$ is a Lyndon word. In particular,
$$\ell^*_n\le \ell_{n,1}+\cdots+\ell_{n,n-1}.$$
\end{lemma}\label{klkl}
\begin{proof} If $w=ub[q-1]^t$ is Lyndon word of length $n$ then the next Lyndon word in Duvel's algorithm is $P(D(ub[q-1]^t))=P(ub[q-1]^t)=u[b+1]$.

If $w=u[q-2]$ is a Lyndon word of length $n$ then $P(D(w))=u[q-1]$ is the next Lyndon word in Duvel's algorithm.
\end{proof}

The amortized number of updates of listing all the Lyndon words of length at most $n$ in Duval's algorithm is~\cite{D88}
$$\gamma_n\le \frac{2\Lambda_n}{L_n}-1=1+\frac{2}{q-1}+O\left(\frac{1}{qn}\right)$$ We now show that
\begin{theorem} Using the compressed representation
the amortized number of updates for enumerating all the Lyndon words of length exactly $n$ is at most
$$\frac{3(\Lambda_n-L_n)+\ell_n}{\ell_n}= 1+\frac{3q}{(q-1)^2}+o(1)$$
\end{theorem}
\begin{proof}
The number of Lyndon words of length $n$ of the forms $w=ub$ where $b\in \Sigma$, $b\not=[q-1]$ and $b\not= [q-2]$ is
$\ell_n-(\ell_{n,1}+\cdots+\ell_{n,n-1})-\ell^*_n$. The next word of length $n$ is $u[b+1]$. So each such word takes one update to find the next word.
For words that end with the symbol $[q-2]$ we need to change this symbol to $[q-1]$ and plausibly merge it with the previous one in the compressed representation. This takes at most two updates. One for removing this symbol and one for merging it with the cells of the form $[q-1]^t$. Therefore for such words we need $2\ell_n^*$ updates. Thus, for Lyndon words that do not ends with $[q-1]$ we need
$\ell_n-(\ell_{n,1}+\cdots+\ell_{n,n-1})+\ell_n^*$ updates.

For strings of the form $w=ub[q-1]^t$ where $b\not=[q-1]$ and $t\ge 1$ we need at most $3t$ updates and therefore at most $3t\ell_{n,t}$ for all such words. See the proof of Theorem~\ref{Th3}. Therefore, the total updates is at most
$$\ell_n-(\ell_{n,1}+\cdots+\ell_{n,n-1})+\ell_n^*+3(\ell_{n,1}+2\ell_{n,2}+\cdots+(n-1)\ell_{n,n-1})$$
By Lemma~\ref{klkl}, this is at most
$$\ell_n+3(\ell_{n-1}+2\ell_{n-2}\cdots+(n-1)\ell_{1}).$$
Now, the amortized update is
\begin{eqnarray*}
\frac{\ell_n+3(\ell_{n-1}+2\ell_{n-2}\cdots+(n-1)\ell_{1})}{\ell_n}
&=&\frac{3(\Lambda_n-L_n)+\ell_n}{\ell_n}\\
&=&1+3\frac{\Lambda_n-L_n}{\ell_n}.\\
\end{eqnarray*}

By (\ref{one}), (\ref{two}) and (\ref{three}) we get
\begin{eqnarray*}
1+3\frac{\Lambda_n-L_n}{\ell_n}&\le & 1+3\frac{\frac{q^2}{(q-1)^2}\left(1+\frac{2}{(q-1)(n-1)}+O\left(\frac{1}{(qn)^2}\right)\right)-\frac{q}{q-1}}
{1-\frac{q}{(q-1)q^{n/2}}}\\
&=& 1+3\frac{\frac{q}{(q-1)^2}+\frac{q^2}{(q-1)^2}\left(\frac{2}{(q-1)(n-1)}+O\left(\frac{1}{(qn)^2}\right)\right)}
{1-\frac{q}{(q-1)q^{n/2}}}\\
&=& 1+\frac{3q}{(q-1)^2}+O\left(\frac{1}{qn}\right).
\end{eqnarray*}
\end{proof}

\section{Membership in $L_{n,q}$}
In this subsection, we study the complexity of deciding membership in $L_{n,q}$.
That is, given a word $\sigma\in \FF_q^n$. Decide whether $\sigma$ is in $L_{n,q}$.

Since $\sigma\in L_{n,q}$ if and only if for all $1<i\le n$, $R_i(\sigma)>\sigma$,
and each comparison of two words of length $n$ takes $O(n)$ operations, membership can be decided in time $O(n^2)$. Duval in \cite{D83} gave a linear time algorithm.
In this subsection, we give a simple algorithm that
decides membership in linear time. To this
end, we need to introduce the suffix tree data structure.

The suffix tree of a word $s$ is a trie that contains all the suffixes of $s$.
See for example the suffix tree of the word $s=1010110\$$ in Figure~\ref{SuffixTree}.
A suffix tree of a word $s$ of length $n$ can be constructed in linear time in $n$~\cite{W73,F97}.
Using the suffix tree, one can check if a word $s'$ of length $|s'|=m$ is a suffix of $s$ in time $O(m)$.

\begin{figure}
\centering
\includegraphics[trim = 0 0cm 0 0cm,width=1.1\textwidth]{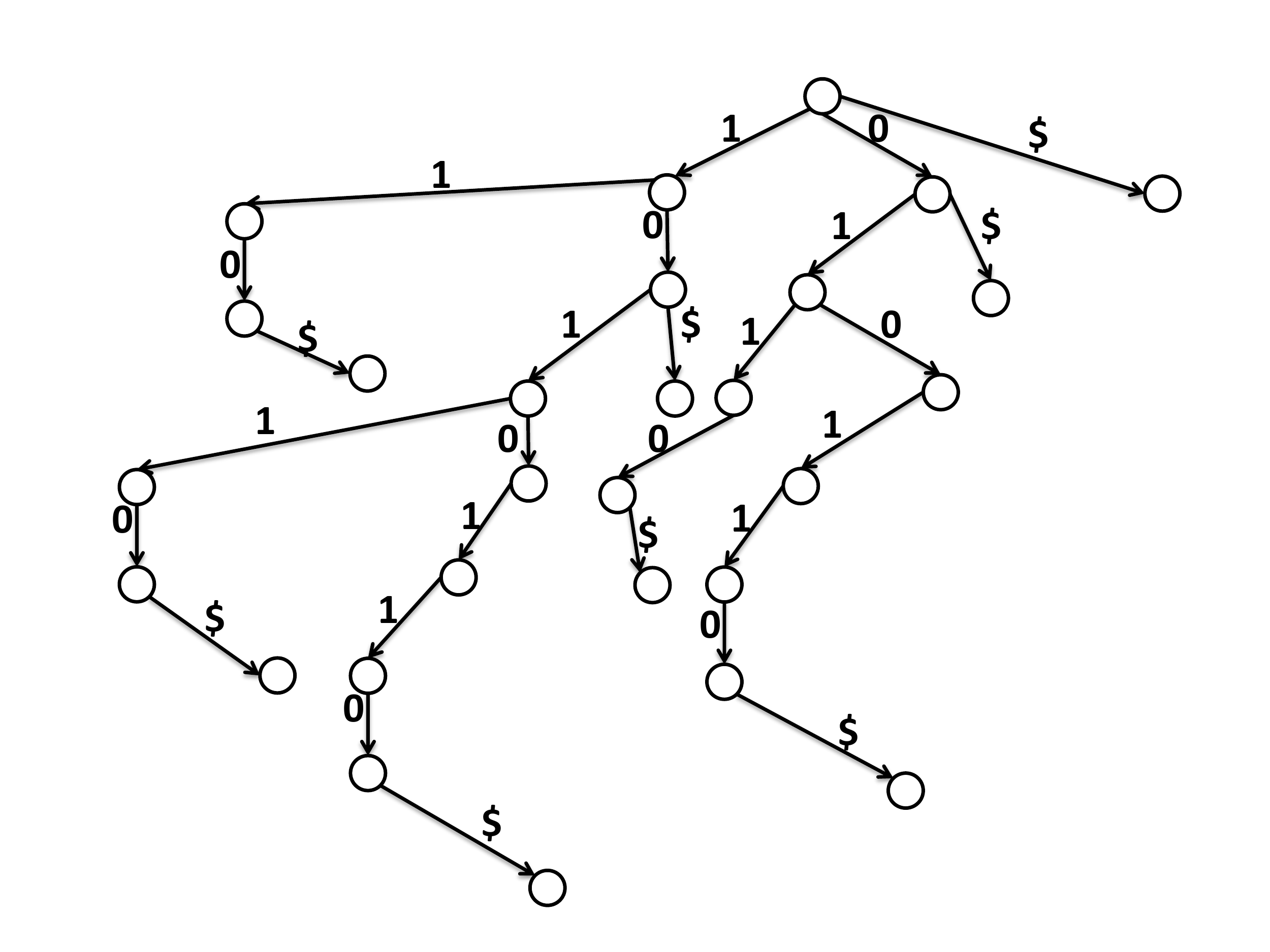}
\caption{The suffix tree of $s=1010110\$$. If $1<0<\$$ then $\Min(ST(s))=110\$$. If $0<1<\$$ then $\Min(ST(s))=010110\$$.
}
\label{SuffixTree}
\end{figure}

Denote by $ST(s)$ the suffix tree of $s$. Define any order $<$ on the symbols of $s$.
Define Min$(ST(s))$ as follows:
Start from the root of the trie
and follow, at each node, the edges with the minimal symbol. Then Min$(ST(s))$ is the word that corresponds to this path.
One can find this word in $ST(s)$ in time that is linear in its length.

The function $\Min$ defines the following total order $\prec$ on the suffixes: Let $T=ST(s)$. Take $\Min(T)$ as the minimum element in that order. Now remove this word from $T$ and take $\Min(T)$ as the next one in that order.
Repeat the above until the tree is empty. For example, if $0<1<\$$ then the order in the suffix tree in Figure~\ref{SuffixTree} is
$$010110\$,0110\$,0\$,1010110\$,10110\$,10\$,110\$,\$. $$
Obviously, for two suffixes $s$ and $r$, $s\prec r$ if and only if for $j=\min(|r|,|s|)$ we have $s_1\cdots s_j<r_1\cdots r_j$ (in the lexicographic order).

We define $ST_m(s)$ the suffix tree of the suffixes of $s$ of length at least $m$.
We can construct $ST_m(s)$ in linear time in $|s|$ by taking a walk in the suffix tree
$ST(s)$ and remove all the words of length less than $m$. In the same way as above, we define
$\Min(ST_m(s))$.

We now show
\begin{lemma}\label{cond01} Let $\$\not\in\FF_q$
be a symbol. Define any total order $<$ on $\Sigma=\FF_q\cup\{\$\}$ such that $\$<\alpha$
for all $\alpha\in \FF_q$.
Let $\sigma\in \FF_q^n$. Then $\sigma\in L_{n,q}$ if and only if
$$\Min(ST_{n+2}(\sigma\sigma\$))=\sigma\sigma\$.$$
\end{lemma}
\begin{proof} First, notice that every word in $ST_{n+2}(\sigma\sigma\$)$ is of the form
$\sigma_i\cdots\sigma_n\sigma\$$ for some $i=1,\ldots,n$. Let $T=ST_{n+2}(\sigma\sigma\$)$.

If $R_i(\sigma)<\sigma$ then
$\sigma_{i}\cdots\sigma_n\sigma_1\cdots\sigma_{i-1} < \sigma$,
and therefore $\sigma_{i}\cdots\sigma_n\sigma\$=\sigma_{i}\cdots\sigma_n\sigma_1\cdots\sigma_{i-1}\sigma_{i}\cdots \sigma_n\$
\prec \sigma\sigma\$$. Thus,
$\Min(T)\not=\sigma\sigma\$$.

If $R_i(\sigma)=\sigma$ then $\sigma_{i}\cdots\sigma_n\sigma_1\cdots\sigma_{i-1} = \sigma$,
and then $$\sigma_{i}\cdots\sigma_n\sigma=\sigma_{i}\cdots\sigma_n\sigma_1\cdots\sigma_{i-1} \sigma_{i}\cdots\sigma_n= \sigma\sigma_1\cdots\sigma_{n-i+1}.$$
Thus, $\sigma_{i}\cdots\sigma_n\sigma\$< \sigma\sigma_1\cdots\sigma_{n-i+2}$ which implies
$\sigma_{i}\cdots\sigma_n\sigma\$\prec \sigma\sigma\$$.
Therefore, we have
$\Min(T)\not=\sigma\sigma\$$.

If $R_i(\sigma)>\sigma$ then $\sigma_{i}\cdots\sigma_n\sigma_1\cdots\sigma_{i-1} > \sigma$,
and therefore $\sigma_{i}\cdots\sigma_n\sigma\$ $ $\succ \sigma\sigma\$$
and then $\Min(T)\not=\sigma_{i}\cdots\sigma_n\sigma\$$.

Now, if $\sigma\in L_{n,q}$ then $R_i(\sigma)>\sigma$ for all $1<i\le n$. Thus $\Min(T)\not=\sigma_{i}\cdots\sigma_n\sigma\$$ for all
$i$. Therefore we have $\Min(T)=\sigma\sigma\$$.
If $\sigma\not\in L_{n,q}$ then there is $i$ such that $R_i(\sigma)\le \sigma$, and then $\Min(T)\not=\sigma\sigma\$$.
\end{proof}

\begin{figure}[h!]
  \begin{center}
  \fbox{\fbox{\begin{minipage}{28em}
  \begin{tabbing}
  xxxx\=xxxx\=xxxx\=xxxx\= \kill
  {\bf Membership$(\sigma,n,q)$}\\
  \> 1) Define a total order on $\FF_q\cup \{\$\}$ such that $\$$ is the\\
  \> \> minimal element.\\
  \> 2) $T\gets$Construct the Suffix Tree of $\sigma\sigma\$$.\\
  \> 3) Take a walk in $T$ and remove all the words of length less\\
  \>\> than $n+2$.\\
  \> 4) Define $r$ the word of the path that start from the root\\
  \> \> and takes, at each node, the edge with the smallest symbol.\\
  \> 5) If $r=\sigma\sigma\$$ then $\sigma\in L_{n,q}$ else $\sigma\not\in L_{n,q}$.\\
  \end{tabbing}
  \end{minipage}}}
  \end{center}
	\caption{Membership of $\sigma$ in $L_{n,q}$.}
	\label{Alg02}
	\end{figure}

We now prove
\begin{theorem}\label{linear} There is a linear time algorithm that decides whether a word $\sigma$ is in $L_{n,q}$.
\end{theorem}
\begin{proof} The algorithm is in Figure~\ref{Alg02}.
We use Lemma~\ref{cond01}. The algorithm constructs the trie $ST_{n+2}(\sigma\sigma\$)$.
The construction takes linear time in $\sigma\sigma\$$ and therefore linear time in $n$.
Finding $\Min(ST_{n+2}(\sigma\sigma\$))$ in a trie takes linear time.
\end{proof}

\end{document}